\RequirePackage{fix-cm}
\documentclass[smallextended]{svjour3mod}       % onecolumn (second format)
\smartqed  % flush right qed marks, e.g. at end of proof
\usepackage{graphicx}
\usepackage{dcolumn}% Align table columns on decimal point
\usepackage{bm}% bold math
\usepackage{verbatim}       % Defines \begin{comment}, \end{comment}

\usepackage{amssymb}

\usepackage{amstext}
\usepackage{amsmath}
\usepackage{amsfonts}
\usepackage{units}
\usepackage{mathrsfs}
\usepackage{cite}

\newcommand{\p}{\partial}

\newcommand{\dd}{{\rm d}}

\begin{document}

\title{The vacuum conservation theorem\thanks{Work partially supported by GNFM of INDAM.}
}


\author{E. Minguzzi
}


\institute{E. Minguzzi \at
              Dipartimento di Matematica e Informatica ``U. Dini'', Universit\`a
degli Studi di Firenze, Via S. Marta 3,  I-50139 Firenze, Italy. \\
              \email{ettore.minguzzi@unifi.it}           %  \\
}
\date{}

\maketitle

\begin{abstract}
A version of the vacuum conservation theorem is proved which does not assume the existence of a time function nor demands stronger properties than the dominant energy condition. However,  it is shown that a stronger stable version plays a role in the study of compact Cauchy horizons.
%\keywords{Causality \and energy conditions \and Cauchy horizons}
%\PACS{04.20.Dw \and 04.20.Dw}
\end{abstract}

\section{Introduction}
Hawking's vacuum conservation theorem  \cite{hawking70b} states that in a spacetime which admits a time function and which satisfies the dominant energy condition, the stress-energy tensor vanishes on a compact domain delimited by two locally achronal hypersurfaces $S_1$ and $S_2$  provided it vanishes on $S_1$ (or $S_2$). Roughly speaking, if a space region is empty at one time it will remain so at later times provided no matter-energy flows in from outside.
%if the stress-energy tensor vanishes on a compact domain $S$ of a spacelike hypersurface, then it vanishes in $D^+(S)$.
A more detailed proof can be found in \cite[Sect.\ 4.3]{hawking73} and a simplified argument was obtained by Carter in \cite{carter02}.

In this work we wish to remove the assumption on the existence of a time function.
%according to which there is a time function $\tau$ in a neighborhood of $\overline{D^+(S)}$.
This is desirable since the domain under consideration is often of the form $\overline{D^+(S_1)}$, with $S_1$ partial Cauchy hypersurface, where  the horizon $S_2:=H^+(S_1)$ might form precisely due to the presence of closed timelike curves behind it, where a time function cannot exist.
\begin{remark}
Hawking was aware of this limitation in his theorem and in fact on \cite[p.\ 298]{hawking73} he sketched a possible more involved proof which however demanded a slightly stronger energy condition \cite[point (5), p.\ 293]{hawking73}.
%The new energy  condition was not  satisfied by some pure radiation fields (Type II \cite[Sect.\ 4.3]{hawking73}) and so it was somewhat unsatisfactory.
Furthermore, in the new argument Hawking still assumes that a time function of the globally hyperbolic set $\textrm{Int}D(S_1)$ extends up to the boundary $H^+(S_1)$ remaining $C^1$, which should not be expected in general due to the lack of regularity of $H^+(S_1)$.
\end{remark}

In this work we shall retain Carter's simplification and we shall reach the desired conclusion without invoking the existence of a time function. Also,  we shall not need to strengthen the  energy condition.

\section{The dominant energy condition}

 We recall that a {\em spacetime} $(M,g)$ is a paracompact, connected, time oriented Lorentzian manifold of dimension $N+1\ge 2$. The metric has signature $(-,+,\cdots,+)$. In our terminology the lightlike, timelike or causal vectors are non-zero. A null (resp.\ non-spacelike) vector is lightlike (resp.\ causal) or zero.
We assume that the metric satisfies the Einstein equations with cosmological constant, and that the stress-energy tensor satisfies:

\begin{definition}
The {\em dominant energy condition}. On every tangent space $T_pM$ the linear endomorphism
\begin{equation} \label{llp}
v^\mu \to -T^\mu_{\ \nu}\, v^\nu
\end{equation}
sends the future-directed non-spacelike cone into itself.
\end{definition}
%By non-spacelike cone we mean the causal cone plus the zero vector.
When $v^\mu$ is timelike and normalized $-T^\mu_{\ \nu}\, v^\nu$ represents the energy-mumentum flow which the condition demands to be causal, roughly speaking energy cannot travel faster than light. Physically, the dominant energy condition is nothing more than that since the constraint for $v$ null follows by continuity from that for $v$ timelike.

Let us recall some basic consequences of the dominant energy condition.
As observed by Carter a non-trivial endomorphism  (\ref{llp}) cannot send a future-directed timelike vector $v$ into the zero vector. Indeed, let us consider a small perturbation
\[
v^\mu+\delta v^\mu \to -T^\mu_{\ \nu} (v^\nu+\delta v^\nu)= -T^\mu_{\ \nu} \,\delta v^\nu.
\]
Observe that since the timelike cone is open $v^\mu+\delta v^\mu$ is timelike for sufficiently small $\delta v^\mu$.
Clearly, if the rank of $T$ is non-zero then we can find a  $\delta v^\nu$ with non-vanishing image, and altering its sign if necessary, we obtain an image vector which is not a future-directed non-spacelike vector. The contradiction proves that only null vectors can be sent to the zero vector.

Since the scalar product of two future-directed causal vectors is non-positive and vanishes if and only if they are proportional and lightlike we see that, for any two future-directed causal vectors $v$ and $w$
\[
T(v,w)\ge 0,
\]
equality holding only if at least  one of the vectors is lightlike. In particular, the dominant energy condition (DEC) implies the {\em weak energy condition} (WEC) which demands that for every timelike vector $v$ (and hence by continuity, for every null $v$) $T(v,v) \ge 0$, and the WEC implies the {\em null energy condition} (NEC) also called {\em null convergence condition} according to which $T(v,v)\ge 0$ for every null vector $v$.

We need to find some inequalities implied by the dominant energy condition \cite{hawking73}.
Observe that if $V^\mu$ is future-directed timelike and normalized and $A^\mu$ is orthogonal to it then, with $\vert A\vert= (A_\alpha A^\alpha)^{1/2}$, the vector $V^\mu\pm A^\mu/ \vert A\vert$ is lightlike thus
\[
T^{\mu \nu} V_\mu (V_\nu\pm A_\nu/ \vert A\vert)\ge 0,
\]
which can be rewritten
\begin{equation} \label{one}
\vert T(V, A) \vert \le \vert A\vert\, T(V,V).
\end{equation}
Similarly, if $B^\mu$ is another vector orthogonal to $V^\mu$,$V^\mu+B^\mu/ \vert B\vert$ is lightlike thus
\[
T^{\mu \nu} (V_\mu+B_\mu/ \vert B\vert) (V_\nu+A_\nu/ \vert A\vert)\ge 0,
\]
\begin{equation} \label{aao}
T(V,V)+T(V,B)/\vert B\vert+ T(V,A)/\vert A\vert+T(A,B) / (\vert A\vert \, \vert B\vert)\ge 0 .
\end{equation}
We are going to obtain an inequality involving just $T(A,A)$, thus we can redefine $A\to -A$ if necessary to make $T(V,A)\le 0$.
Replacing $B=A$ and $B=-A$ in Eq.\
%With $B=-A$ we obtain
%\[
%T(A,A) \le T(V,V) \,\vert A\vert^2 ,
%\]
%with $B=A$, and replacing $A\to -A$ if necessary, we can assume $T(V,A)\le 0$ which using
(\ref{aao}) gives
\[
\vert T(A,A)\vert \le T(V,V) \, \vert A\vert^2 .
\]
From this inequality we obtain the two inequalities
\begin{align*}
 T(\frac{A}{\sqrt{2}\vert A\vert}+\frac{B}{\sqrt{2}\vert B\vert},\frac{A}{\sqrt{2}\vert A\vert}+\frac{B}{\sqrt{2}\vert B\vert}) & \le T(V,V) \, (1+A^\mu B_\mu/ (\vert A\vert \, \vert B\vert)) ,\\
 -T(\frac{A}{\sqrt{2}\vert A\vert}-\frac{B}{\sqrt{2}\vert B\vert},\frac{A}{\sqrt{2}\vert A\vert}-\frac{B}{\sqrt{2}\vert B\vert}) & \le T(V,V) \, (1-A^\mu B_\mu/ (\vert A\vert \, \vert B\vert)) .
\end{align*}
Summing (and considering the equation so obtained also for $B\to -B$)
\begin{equation} \label{two}
\vert T(A,B)\vert \le T(V,V) \, \vert A\vert \, \vert B\vert .
\end{equation}
Any tensor $D_{\mu\nu}$ whose contractions with $V^\alpha$ vanishes can be written as the sum of $N^2$ terms of the form $A_\mu B_\nu$ thus
\begin{equation} \label{three}
\vert T^{\mu \nu} D_{\mu \nu} \vert \le  N^2 \sqrt{{D_{\alpha \beta} D^{\alpha \beta}}} \, T(V,V).
\end{equation}
Equations (\ref{one}), (\ref{two}), prove that in any orthonormal frame in which $e_0$ is timelike
\[
\vert T^{ab}\vert \le T^{00}, \qquad a,b=0,1,\cdots, N.
\]
%%Let us consider the physical spacetime case $N=3$.
%
%%\newpage
%The classification of the symmetric tensors in Minkowski space \cite{collinson72,hawking73,hall76,hall82,gohberg05} states that we can always find  an orthonormal base $\{e_a\}$ with $e_0$ future-directed and timelike in such a way that $T^{ab}$ takes one of the following forms
%\begin{align*}
%& \textrm{Type I:} \ {\small \begin{pmatrix}
%\mu &  & &  \\
% & p_1 & &\\
% & & p_2 &\\
% & & & p_3
%\end{pmatrix}}, \qquad \textrm{Type II:}\ {\small  \begin{pmatrix}
%\alpha+\kappa & \alpha  & &  \\
% \alpha & \alpha - \kappa & & \\
% & & p_1 &\\
% & & & p_2
%\end{pmatrix}}, \alpha=\pm 1\\
%&\\
%& \textrm{Type III:} \ {\small \begin{pmatrix}
%\nu & 0 &  1 &  \\
% 0 & -\nu & 1 & \\
% 1 &  1&  -\nu &\\
% & & & p_1
%\end{pmatrix}}, \qquad \textrm{Type IV:} \ {\small \begin{pmatrix}
%0 & \nu  &   &  \\
% \nu & -\kappa &  & \\
%  &  &  -\nu &\\
% & & & p_1
%\end{pmatrix}}, \vert \kappa\vert <2 \vert \nu\vert
%\end{align*}
%
%The following useful fact shows that in general relativity we can ignore the last two types.
%\begin{proposition}
%The stress-energy tensors of types III or IV violate the null energy condition.
%\end{proposition}
%
%\begin{proof}
%To prove the claim for Type III observe that $T(v,v)=-2<0$ for the lightlike vector  $v^a=(1,0,1,0)$. For Type IV observe that $T(v,v)=-2\vert \nu\vert-k<0$ for $v^a= (1, \textrm{sgn}( \nu) , 0, 0)$.
%\end{proof}
%
%\newpage

\section{Vacuum conservation theorem}

We are ready to prove a first vacuum conservation result.

\begin{theorem} \label{kaq}
Let $(M,g)$ be a spacetime which satisfies the dominant energy condition. Let us consider an open relatively compact region $U$ of spacetime delimited by two  locally achronal hypersurfaces $S_1$ and $S_2$, possibly with edge, such that $\bar{U}$ is the union of integral connected segments (possibly degenerate to a point) of a $C^1$ future-directed timelike vector field $V$ which have past endpoint in  $S_1$ and future endpoint in $S_2$.
%
%
%(a) $\p S_1=\p S_2$,
%%(b)  $S_1$ or $S_2$ is $C^1$,
% and (b) there is a $C^1$ global future-directed timelike vector field $V$ such that every inextendible future-directed integral curve of $V$  passing through some point of $\bar{U}$ meets $S_1$ and then  $S_2$.
Then the stress energy tensor vanishes on $S_1$ iff it vanishes on $S_2$ iff it vanishes on $\bar U$.
\end{theorem}

\begin{proof}
%We can assume without loss of generality that $S_1$ is $C^1$ since it necessary we can redefine the orientation of the spacetime.

%otherwise we have just to change the time orientation of the spacetime.

%We can rescale $V$ through multiplication by a function going sufficiently fast to zero at infinity so as to make it complete.
Let us first give the proof for $S_1$ and $S_2$ continuously differentiable.
Let us normalize $V$, $V^\alpha V_\alpha=-1$.
Let  $(t,p) \to \varphi_t(p)$  be the map which assigns to $(t,p)$ the endpoint of the curve $x:[0,t] \to M$, $x(0)=p$, $\dot{x}=V$, wherever it exists. By well known results on the dependence on initial conditions and external parameters of ordinary differential equations \cite{hartman64} this flow map $\varphi$ is $C^1$ on an open subset of $\mathbb{R}\times M$ where the map $(t,p) \to \varphi_t(p)$  makes sense. For every $t$ the map $\varphi_t$ is actually a local diffeomorphism wherever it is defined \cite[Theor.\ 2.9]{lang95}.

%Let $\varphi_t$ be the smooth flow map of the vector field $V$ and let it be  defined on the subset of $\mathbb{R}\times M$ wherever the map $(t,p) \to \varphi_t(p)$  makes sense.
By assumption every point of $\bar{U}$ can be reached from $S_1$ following an integral line of $V$.
Let $t\colon \bar U \to [0,+\infty)$ be the function defined through: $t(q)$ is such that  $\varphi_t(x) =q$ for some $x\in S_1$. By the implicit function theorem $t$ has the same regularity of $S_1$ thus $C^1$, and by construction $\dd t[V]=1$. We stress that we are not claiming that $t$ is a time function, just that it is well defined over $\bar{U}$.

For  a real number $C$ to be specified later on, let us introduce the  vector field
\[
X^\mu= - e^{-C t} T^\mu_{\ \nu} V^\nu ,
\]
which is future-directed timelike wherever $T\ne 0$ or zero. Let us denote with $\mu$  the usual volume form.
%Let us denote with $\hat S_k$ the open subset of $S_k$ in which $X\ne 0$.
%Let us endow $\hat S_k$, $k=1,2$, with the following volume form
%\[
%\nu=\frac{1}{-g(X,n)}\, i_X  \mu
%\]
%where $n$ is a $C^0$ future-directed causal  field normal to $S_k$ (possibly lightlike). Here $\nu$ is evaluated just on the tangent space to $\hat S_{i}$. This choice of volume is independent of the transverse field $X$ but it depends on the scale of  $n$ (this multiplying factor could be naturally fixed through normalization if $n$ were  timelike).
Let us consider the non-negative integrals
\begin{align*}
I_{k}=\int_{S_{k}} i_X  \mu , \qquad k=1,2.
%=\int_{\hat S_{k}} i_X  \mu= \int_{\hat S_{k}} [-g(X,n)] \nu.
\end{align*}
If $X\ne 0$ at a point $r\in S_k$, $k=1,2$, then $X$ is timelike and hence transverse to $S_k$, thus the integral is strictly positive. As a consequence, for any $k=1,2$, $I_k=0$ iff $X=0$ on $S_k$ iff $T=0$ on $S_k$.

%Clearly, $I_k=0$ iff $X=0$ on  $S_k$ iff $T=0$ on $S_k$.
Recalling the identity $d\, i_X \mu = \textrm{div} X \mu$ we obtain from Stokes theorem
\[
I_2-I_1=\int_U \textrm{div} X \mu.
\]
Using $T^{\nu \mu}_{\ ;\mu}=0$
\begin{equation} \label{div}
\textrm{div} X=e^{-Ct} T^{\mu \nu} [C t_{;\mu} V_\nu-V_{\mu; \nu}] .
\end{equation}
Since
\[
[C t_{;\mu} V_\nu-V_{\mu; \nu}] V^\mu V^\nu=-C ,
\]
the matrix in square brackets can be written
\begin{equation} \label{dec}
C t_{;\mu} V_\nu-V_{\mu; \nu}=-C \, V_\mu V_\nu+V_\mu A_\nu+ B_\mu V_{\nu} +D_{\mu \nu} ,
\end{equation}
for some continuous tensor fields $A, B, D$  such that their contractions  with $V$ vanish. This fact can be most easily understood passing to an orthonormal basis with $e_0= V$. In the same way one easily sees that $A^\alpha A_\alpha, B^\beta B_\beta\ge 0$ and $D^{\mu \nu}D_{\mu \nu}\ge 0$.

On every relatively compact open set $O$ which admits a coordinate system the contraction $T^{\mu \nu} [C t_{;\mu} V_\nu-V_{\mu; \nu}]$ can be made non-positive by choosing a sufficiently large value of $C$, and non-negative by choosing a sufficiently large value of $-C$.  Indeed, by the dominant energy condition and by using the continuity of $A, B,  V,$ on the compact set $\bar{O}$ we have that, due to Eqs.\ (\ref{one}) and (\ref{three}), there is a constant $K>0$ such that ($N+1$ is the spacetime dimension)
\begin{align*}
\vert T^{\mu \nu} V_\mu A_\nu \vert &\le \vert A\vert \, T(V, V) \le K\, T(V,V),\\
\vert T^{\mu \nu} D_{\mu \nu}\vert &\le N^2 \sqrt{{D_{\alpha \beta} D^{\alpha \beta}}} \,  T(V,V)\le K \, T(V,V),
\end{align*}
where an equation analogous to the first one holds also for $B$.
The statement at the beginning of the paragraph  follows plugging (\ref{dec}) into Eq.\ (\ref{div}).

Since $\bar{U}$ can be covered by a finite number of relatively compact coordinated neighborhoods we have on $\bar{U}$, $(\textrm{sgn} C )\textrm{div} X \le 0$. Observe that actually by choosing $\vert C\vert$ sufficiently large the equality case is excluded unless $T=0$ at the point under consideration.
Thus for $C>0$ large enough
\[
0\le I_2 \le I_1
\]
where the second equality holds only if $T=0$ on $U$ and hence $\bar{U}$.
Thus if we take $C$ large enough we have that $T=0$ on $S_1$ implies  $T=0$ on $\bar U$ and hence in $S_2$.

Analogously, for $-C>0$ large enough
\[
0\le I_1\le I_2
\]
where the second equality holds only if $T=0$ on $U$ and hence $\bar{U}$.
Thus if we take $-C$ large enough we have that $T=0$ on $S_2$ implies  $T=0$ on $\bar U$ and hence in $S_1$.

The proof for the Lipschitz regularity of $S_1$ and $S_2$ is as above, there are just two critical steps. The first is the construction of function $t$. We cannot define it starting the flow from $S_1$ since, as this hypersurface is only Lipschitz, $t$ would not be $C^1$. However, in a neighborhood of $p \in S_1$ one can introduce coordinates $\{ x^\mu\}$ such that $\p_0=V$, and $\{x^i\}$ are induced from the $V$-projection on a smooth spacelike hypersurface $\Sigma$  transverse to $V$. Then $S_1$ can be locally expressed as a Lipschitz graph $x^0=h({\bf x})$ which can be approximated by a smooth graph by Whitney approximation theorem \cite[Theor.\ 6.21]{lee03}. Through a partition of unity it is possible to patch together these local graphs into a smooth hypersurface $S_1'$ with the property that every $V$-integral curve which intersect $S_1$ intersects $S_1'$ and conversely. Then the definition of $t$ is as above but starting the flow from $S_1'$.

The other delicate step concerns the application of the  Gauss-Green (divergence) theorem which, fortunately, has indeed been generalized to Lipschitz domains and even Lipschitz fields \cite[Sect.\ 5.8]{evans92} \cite{pfeffer12}. In this theorem the boundary term $\int_B i_X\mu$, $B\subset S_k$, can be calculated expressing $S_k$ as a local graph and extending the expression of the integral in terms of a $C^1$ function $h$ to the Lipschitz case. More in detail, if $X=a(x^0,{\bf x}) \p_0+ b^i(x^0,{\bf x}) \p_i$
\begin{align*}
i_X\mu \vert_{S_k}&:=\sqrt{-\vert g\vert (h({\bf x}), {\bf x})} \, i_X \dd x^0\wedge \cdots\wedge \dd x^{N}\vert_{S_k}\\
&=\sqrt{-\vert g\vert (h({\bf x}), {\bf x})}\, (a-\p_b h) \, \dd x^1\wedge \cdots\wedge \dd x^{N} .
\end{align*}
Observe that $\p_b h$ is $L^1_{loc}$ by the Lipschitzness of $h$, thus the integral $\int_B i_X\mu$ is well defined. $\square$
\end{proof}

\begin{remark}
The physical meaning of the integral $I_k$ can be understood, at least in the continuously differentiable case, as follows.
Let us denote with $\hat S_k$, $k=1,2$, the open subset of $S_k$ in which $X$ is timelike.
Let us endow $\hat S_k$ with the following volume form
\[
\nu=\frac{1}{-g(X,n)}\, i_X  \mu
\]
where $n$ is a $C^0$ future-directed causal  field normal to $S_k$ (possibly lightlike). Here $\nu$ is evaluated just on the tangent space to $\hat S_{i}$. This choice of volume is independent of the transverse field $X$ but it depends on the scale of  $n$ (this multiplying factor could be naturally fixed through normalization if $n$ were  timelike). We have
\begin{align}
I_{k}=\int_{S_{k}} i_X  \mu
=\int_{\hat S_{k}} i_X  \mu= \int_{\hat S_{k}} [-g(X,n)] \nu= \int_{\hat S_{k}} e^{-Ct} T(V,n) \nu . \label{kos}
\end{align}
One should be very careful in using the last expression when $S_k$ is partly lightlike for one can easily miss the difference between $\hat S_k$ and $S_k$.
\end{remark}

\begin{corollary} \label{odp}
Let $S$ be a  locally achronal compact topological hypersurface possibly with edge and let $\overline{D^+(S)}$ be compact. Under the dominant energy condition the stress energy tensor vanishes on $S$ iff it vanishes on $H^+(S)$ iff it vanishes on $\overline{D^+(S)}$.
\end{corollary}

\begin{proof}
Let $V$ be any smooth normalized future-directed timelike field, and
 let $S_1:=S$ and $S_2:=H^+(S)$. The assumption of Theorem \ref{kaq} are satisfied thus the desired conclusion follows. $\square$
\end{proof}

%and let $p\in \textrm{Int} D^+(S)$. Since $\textrm{Int} D(S)$ is globally hyperbolic there is a $C^1$ time function $t: \textrm{Int} D^+(S)\to \mathbb{R}$ whose level sets have the same properties of the set $S_2$ in Theorem \cite{kaq}. Let us define $S_2=\{q\in D^+(S): t(q)=t(p)\}$. From that Theorem $T=0$ on $S_1$ implies $T(p)=0$ and since $p$ is arbitrary $T=0$ on  $\overline{D^+(S)}$.
%
%
%If $p\in H^{+}(S)$ define $S_2:=H^+(S)$ otherwise observe that there is a time function $t: D^+(S)\to \mathbb{R}$ whose level sets have the same properties of the set $S_2$ in Theorem \cite{kaq} and define $S_2=\{q\in D^+(S): t(q)=t(p)\}$. From that Theorem $T=0$ on $S_1$ implies $T(p)=0$ and since $p$ is arbitrary $T=0$ on $\overline{D^+(S)}$. The other directions are trivial or follow from Theorem \cite{kaq}.

The usual vacuum conservation theorem states that if the stress-energy tensor vanishes on $S$ then it vanishes on $D^+(S)$ and on the horizon $H^+(S)$. Actually, the previous result establishes a converse implication  which is of particular interest, namely that if $T\ne 0$ somewhere on $S$ then $T\ne 0$ at some point  $p\in H^+(S)$. Suppose that $H^+(S)$ is $C^1$ at $p$ and that $n$ is a future directed lightlike vector orthogonal to $T_pH^+(S)$.  Observe that we cannot conclude that $T(n,n)\ne 0$ at $p$, for it can be that $T^\mu_{\ \nu} \,n^\nu\propto n^\mu$ as for some stress-energy tensor of pure aligned radiation (Type II \cite[Sect.\ 4.3]{hawking73}).

Let us introduce two modifications of the dominant energy condition.
\begin{definition}
The {\em stable dominant energy condition} \cite{hall82}. On every tangent space $T_pM$ the linear endomorphism  (\ref{llp}) sends the future-directed casual cone into the future-directed timelike cone. \\
The {\em weakened stable dominant energy condition}. At every point $T=0$ or the stable dominant energy condition holds, which is equivalent to: the linear endomorphism  (\ref{llp}) sends the future-directed casual cone into the future-directed timelike cone plus the zero vector.\footnote{That the former characterization implies the latter is clear. For the converse, if at the given point $T= 0$  or if no causal vector is sent to zero we have finished. Thus assume $T\ne 0$ and there is a f.d.\ causal vector $w$ which has zero image. Since $T\ne 0$
there is a f.d.\ causal vector $v$ such that  $u^b:=-T^b_{\ a} v^a\ne 0$, necessarily f.d.\ timelike by assumption, then $0=T(w,v)=-g(w,u)$ which is a contradiction since it must be positive.}
\end{definition}

The stable dominant energy condition  is preserved under small perturbations of the endomorphism, which means that the source content of spacetime is not on the verge of violating the dominant energy condition. Unfortunately, it might seem a too strong condition as it excludes the vacuum case, that is, under this condition $T\ne 0$ at every point. Observe that we might want to impose the stable dominant energy condition only where a source is really present. This leads us to the
weakened stable dominant energy condition which contemplates the vacuum case.

In the four dimensional spacetime case and in the classification of \cite[Sect.\ 4.3]{hawking73}, this condition allows only Type I stress-energy tensors  (namely the diagonal ones) where the energy density is larger than the absolute values of the principal pressures, $\vert p_i\vert<\rho$ or $p_i=\rho=0$ (observe that types III and IV are excluded by the null energy condition).

%Using this classification, it is then easy to check that this condition is exactly Hawking's strengthened  dominant energy condition of \cite[p.\ 293]{hawking73}, according to which the dominant energy condition holds and for any causal vector $v$, $T(v,v)=0$ implies $T^a_{ \ b} v^b=0$.

%Clearly, under the weakened stable dominant energy condition if $T\ne 0$ at a point of  differentiability of the horizon, and $n$ is a  we can conclude that $T(n,n)\ne 0$
Let $p\in H^+(S)$ be a differentiability point of the horizon, and let $n$ be a future directed lightlike vector tangent to the horizon at $p$. If the weakened stable dominant energy condition holds and $T\ne 0$ at $p$ then $T(n,n)\ne 0$ at $p$, thus\footnote{Hawking \cite[p.\ 293]{hawking73} claims a similar result but his proof and claim seem incorrect. He assumes a weaker form of energy condition, which does not exclude the possibility $T^{ab}=\alpha n^a n^b$, where $n$ is a lightlike vector, and then he applies a version of the conservation theorem which he has not really proved. Probably he used Eq.\ (\ref{kos}) missing the difference between $\hat S$ and $S$, and that between $V$ and $n$.}
%the implication $T\ne 0 \Rightarrow T(n,n)\ne 0$ on the horizon holds true, where $n$ is any future directed lightlike vector tangent to the horizon thus
\begin{proposition}
Under the assumptions of Corollary \ref{odp}, and under the weakened stable dominant energy condition   $T\ne 0$ somewhere on $S$ implies $T(n,n)\ne 0$ somewhere on the horizon $H^+(S)$ (namely at some differentiability point of the horizon where a lightlike tangent $n$ to the horizon exists).
\end{proposition}

\begin{proof}
We know that $T\ne 0$ at some point $q\in H^+(S)$, thus in a neighborhood $U$ of $q$. As the horizon is almost everywhere differentiable and $T$ is continuous we can find $p \in U\cap H^+(S)$ where the horizon is differentiable and $T\ne 0$. $\square$
\end{proof}

This result will allow us to show that under the weakened stable dominant energy condition compact Cauchy horizons cannot form \cite{minguzzi14d}, a result which seems significative for a final resolution of the weak/strong cosmic censorship conjecture.

Physically speaking, in absence of quantum effects the weakened stable dominant energy condition seems reasonable in all those phenomena in which it is known that the source of spacetime is not just coherent radiation, particularly in the study of the development of future singularities. Less convincing is its validity at the beginning of the Universe since the Higgs field may still stay in its non-degenerate minimum so that, as all the particles would be massless, the only source content would  be  radiation. In this case if $S$ is a present time partial Cauchy hypersurface the radiation would emerge from the boundary $H^{-}(S)$ which would play the role of Big Bang null hypersurface. For more on this picture of the beginning of the Universe and its advantages the reader is referred to \cite{minguzzi09d,minguzzi10d}.

\section{Conclusions}
We have proved a version of the vacuum conservation theorem which has some good features: (a) it involves causal rather than just spacelike hypersurfaces, (b) it holds for Lipschitz hypersurfaces, (c) it does not assume the existence of a time function, (d) it includes a converse preservation claim according to which if the stress-energy tensor does not vanish on $S$ then it does not vanish on $H(S)$. The application of this theorem to the study of compact Cauchy horizons has been briefly commented.

%\begin{acknowledgements}
%\end{acknowledgements}

%\bibliography{../../bibliografie/simultaneity,../../bibliografie/libri,../../bibliografie/miei,../../bibliografie/mieiPreprints,../../bibliografie/mieiProceedings}


\def\cprime{$'$}

\end{document}